%% file: main.tex
\documentclass[a4paper,11pt]{easychair}
\input{preamble}

\begin{document}
\title{Raising The Bar For \VC: Fixed-parameter Tractability Above
  A Higher Guarantee}
%\subtitle{Faster Fixed-parameter Algorithms Above A Stronger Guarantee}
\titlerunning{\VC parameterized above a higher guarantee}

%%%%%%%%%%%%%%%%% ACHTUNG. Comment beginnt
\iffalse
\author[1]{Shivam Garg}
\author[2]{Geevarghese Philip} 
\authorrunning{S.Garg and G. Philip}

\affil[1]{IIT Bombay, Mumbai, India\\ 
  \texttt{shivamgarg@iitb.ac.in}}
\affil[2]{Max Planck Institute for Informatics, Saarbr\"{u}cken,
  Germany\\ 
  \texttt{gphilip@mpi-inf.mpg.de}}

\fi
%%%%%%%%%%%%%%%%% Endo of ACHTUNG

\author{%
  Shivam Garg\inst{1} \and Geevarghese Philip\inst{2}}

\authorrunning{S.Garg and G. Philip}

\institute{%
  IIT Bombay, Mumbai, India, \email{shivamgarg@iitb.ac.in} \and
  Max Planck Institute for Informatics, Saarbr\"{u}cken, Germany,
  \email{gphilip@mpi-inf.mpg.de}}

% Autoref names. The renewcommands don't take effect if they are
% placed before this point, for some weird reason. So placing all
% of them together, here.
  
%\newcommand{\sectionautorefname}{Section}
%\newcommand{\subsectionautorefname}{Section}
\newcommand{\observationautorefname}{Observation}
\newcommand{\propositionautorefname}{Proposition}
\newcommand{\reductionautorefname}{Reduction Rule}
\newcommand{\branchingautorefname}{Branching Rule}
\newcommand{\claimautorefname}{Claim}
 
\makeatletter
\@ifundefined{algorithmautorefname}{% if
\newcommand{\algorithmautorefname}{Algorithm}}{%else 
\renewcommand{\algorithmautorefname}{Algorithm}}
\makeatother

% So that the first page does not have a page number, just like
% the other pages.
%\makeatletter
%  \let\ps@plain\ps@empty
%\makeatother

%\clearpage
\pagenumbering{gobble} % Remove page numbers (and reset to 1).
\clearpage
\thispagestyle{empty}
\maketitle
\begin{abstract}
  %%%%% ACHTUNG! Comment beginnt
  \iffalse
  We investigate the following above-guarantee parameterization of
  the classical \VC problem: Given a graph \(G\) and
  \(\hat{k}\in\mathbb{N}\) as input, does \(G\) have a vertex
  cover of size at most \((2LP-MM)+\hat{k}\)? Here \(MM\) is the
  size of a maximum matching of \(G\), \(LP\) is the value of an
  optimum solution to the relaxed (standard) LP for \VC on \(G\),
  and \(\hat{k}\) is the parameter. Since
  \((2LP-MM)\geq{LP}\geq{MM}\), this is a stricter
  parameterization than those---namely, above-\(MM\), and
  above-\(LP\)---which have been studied so far.
  
  We prove that \VC is fixed-parameter tractable for this stricter
  parameter \(\hat{k}\): We derive an algorithm which solves \VC
  in time \(\OhStar(3^{\hat{k}})\), pushing the envelope further
  on the parameterized tractability of \VC.
  \fi
  %%%%% ACHTUNG! Comment endet

  The standard parameterization of the \VC problem (Given an
  undirected graph \(G\) and \(k\in\mathbb{N}\) as input, does
  \(G\) have a vertex cover of size at most \(k\)?) has the
  solution size \(k\) as the parameter.  The following more
  challenging parameterization of \VC stems from the observation
  that the size \(MM\) of a \emph{maximum matching} of \(G\)
  lower-bounds the size of any vertex cover of \(G\): Does \(G\)
  have a vertex cover of size at most \(MM+k_{\mu}\)? The
  parameter is the excess \(k_{\mu}\) of the solution size over
  the lower bound \(MM\). 

  Razgon and O'Sullivan (ICALP 2008) showed that this
  \emph{above-guarantee} parameterization of \VC is
  fixed-parameter tractable and can be solved\footnote{The \OhStar
    notation hides polynomial factors.} in time
  \(\OhStar(15^{k_{\mu}})\). This was first improved to
  \(\OhStar(9^{k_{\mu}})\) (Raman et al., ESA 2011), then to
  \(\OhStar(4^{k_{\mu}})\) (Cygan et al., IPEC 2011, TOCT 2013),
  then to \(\OhStar(2.618^{k_{\mu}})\) (Narayanaswamy et al.,
  STACS 2012) and finally to the current best bound
  \(\OhStar(2.3146^{k_{\mu}})\) (Lokshtanov et al., TALG
  2014). The last two bounds were in fact proven for a different
  parameter: namely, the excess \(k_{\lambda}\) of the solution
  size over \(LP\), the value of the \emph{linear programming
    relaxation} of the standard LP formulation of \VC. Since
  \(LP\geq{}MM\) for any graph, we have that
  \(k_{\lambda}\leq{}k_{\mu}\) for \YES instances. This is thus a
  \emph{stricter} parameterization---the new parameter is, in
  general, smaller---and the running times carry over directly to
  the parameter \(k_{\mu}\).
  
  We investigate an even stricter parameterization of \VC, namely
  the excess \(\hat{k}\) of the solution size over the quantity
  \((2LP-MM)\). We ask: Given a graph \(G\) and
  \(\hat{k}\in\mathbb{N}\) as input, does \(G\) have a vertex
  cover of size at most \((2LP-MM)+\hat{k}\)? The parameter is
  \(\hat{k}\). It can be shown % (following Lov\'{a}sz and Plummer,
  % 1986)
  that \((2LP-MM)\) is a lower bound on vertex cover size, and
  since \(LP\geq{}MM\) we have that \((2LP-MM)\geq{}LP\), and
  hence that \(\hat{k}\leq{}k_{\lambda}\) holds for \YES
  instances. Further, \((k_{\lambda}-\hat{k})\) could be as large
  as \((LP-MM)\) and---to the best of our knowledge---this
  difference cannot be expressed as a function of \(k_{\lambda}\)
  alone. These facts motivate and justify our choice of parameter:
  this is indeed a stricter parameterization whose tractability
  does not follow directly from known results.
  
  We show that \VC is fixed-parameter tractable for this stricter
  parameter \(\hat{k}\): We derive an algorithm which solves \VC
  in time \(\OhStar(3^{\hat{k}})\), thus pushing the envelope
  further on the parameterized tractability of \VC.
\end{abstract}
\clearpage
\pagenumbering{arabic} % Restart page numbering at 1.
\input{introduction.tex}
\input{preliminaries.tex}

\input{algorithm.tex}
\input{conclusion.tex}
\clearpage
\pagenumbering{roman} % Restart page numbering at i for the bibliography
\bibliographystyle{plain}
\bibliography{VCAG}
\clearpage
\input{appendix.tex}
\end{document}

%% file: preamble.tex
\usepackage{fullpage}
\usepackage{charter}
\usepackage{microtype}

\usepackage{cite}
\usepackage{verbatim}
\usepackage{amsthm}
\usepackage{amsmath,amstext,amssymb,amsfonts}
\usepackage{hyperref}

\usepackage{algorithm}
\usepackage[noend]{algpseudocode}

\usepackage{xspace}
\usepackage{boxedminipage}
\usepackage[textwidth=0.8in,textsize=small]{todonotes}

\theoremstyle{plain}
\newtheorem{claim}{Claim}
\theoremstyle{definition}
\newtheorem{reduction}{Reduction Rule}
\newtheorem{branching}{Branching Rule}

\newcommand{\YES}{\textsc{Yes}\xspace}
\newcommand{\NO}{\textsc{No}\xspace}
\newcommand{\Oh}{\ensuremath{\mathcal{O}}\xspace}
\newcommand{\OhStar}{\ensuremath{\mathcal{O}^{\star}}\xspace}

\newcommand{\FPT}{\textrm{\textup{FPT}}\xspace}

\newcommand{\NPC}{\textrm{\textup{NP-complete}}\xspace}

\newcommand{\name}[1]{\textsc{#1}}

\newcommand{\VC}{\name{Vertex Cover}\xspace}
\newcommand{\AGVC}{\name{Above-Guarantee Vertex Cover}\xspace}
\newcommand{\agvc}{\name{AGVC}\xspace}

\newcommand{\VCAL}{\name{Vertex Cover Above LP}\xspace}
\newcommand{\vcal}{\name{VCAL}\xspace}
\newcommand{\VCALP}{\name{Vertex Cover Above Lov\'{a}sz-Plummer}\xspace}
\newcommand{\vcalp}{\name{VCAL-P}\xspace}

\newcommand{\surplus}[1]{\ensuremath{\mathbf{surplus}({#1})}}

\newcommand{\defparproblem}[4]{
  \vspace{3mm}
\noindent\fbox{
  \begin{minipage}{.97\textwidth}
  \begin{tabular*}{\textwidth}{@{\extracolsep{\fill}}lr} \textsc{#1}  & {\bf{Parameter:}} #3 \\ \end{tabular*}
  {\bf{Input:}} #2  \\
  {\bf{Question:}} #4
  \end{minipage}
  }
  \vspace{2mm}
}

\theoremstyle{plain} % Theorems, lemmas, propositions
\newtheorem{theorem}{Theorem}
\newtheorem{corollary}{Corollary}

\newtheorem{lemma}{Lemma}

\theoremstyle{remark} % Remarks, notes

\theoremstyle{definition} % For definitions
\newtheorem{definition}{Definition}

%% file: introduction.tex
\section{Introduction}\label{sec:introduction}
The input to the \VC problem consists of an undirected graph \(G\)
and an integer \(k\), and the question is whether \(G\) has a
\emph{vertex cover}---a subset \(S\subseteq{}V(G)\) of vertices
such that every edge in \(G\) has at least one end-point in
\(S\)---of size at most \(k\). This problem is among Karp's
original list of 21 \NPC problems~\cite{Karp1972}; it is also one
of the best-studied problems in the field of parameterized
algorithms and
complexity~\cite{ChenKanjXia2010,Lampis2011,MishraRamanSaurabhSikdarSubramanian2011,RamanRamanujanSaurabh2011}.

The input to a parameterized version of a classical decision
problem consists of two parts: the classical input and a specified
\emph{parameter}, usually an integer, usually denoted by the
letter \(k\). A fixed-parameter tractable (\FPT)
algorithm~\cite{DowneyFellows2013,FlumGroheBook} for this
parameterized problem is one which solves the underlying decision
problem in time \(f(k)\cdot{}n^{c}\) where (i) \(f\) is a
computable function of \(k\) alone, (ii) \(n\) is the size of the
(classical) input instance, and (iii) \(c\) is a constant
independent of \(n\) and \(k\). This running time is often written
as \(\OhStar(f(k))\); the \OhStar notation hides constant-degree
polynomial factors in the running time. A parameterized algorithm
which has an \FPT algorithm is itself said to be (in) \FPT.

The ``standard'' parameter for \VC is the number \(k\) which comes
as part of the input and represents the \emph{size of the vertex
  cover} for which we look---hence referred to, loosely, as the
``solution size''. This is the most extensively studied
parameterization of
\VC~\cite{BalasubramanianFellowsRaman1998,ChandranGrandoni2005,ChenKanjXia2010,Lampis2011,NiedermeierRossmanith1999}. Starting
with a simple two-way branching algorithm (folklore) which solves
the problem in time \(\OhStar(2^{k})\) and serves as the \emph{de
  facto} introduction-cum-elevator-pitch to the field, a number of
\FPT algorithms with improved running times have been found for
\VC, the current fastest of which solves the problem in
\(\OhStar(1.2738^{k})\) time~\cite{ChenKanjXia2010}. It is also
known that unless the Exponential Time Hypothesis (ETH) fails,
there is no algorithm which solves \VC in \(\OhStar(2^{o(k)})\)
time~\cite{ImpagliazzoPaturiZane2001}.

This last point hints at a fundamental drawback of this
parameterization of \VC: In cases where the size of a smallest
vertex cover (called the \emph{vertex cover number}) of the input
graph is ``large''---say, \(\Omega(n)\) where \(n\) is the number
of vertices in the graph---one cannot hope to use these algorithms
to find a smallest vertex cover ``fast''. This is, for instance,
the case when the input graph \(G\) has a large \emph{matching},
which is a set of edges of \(G\) no two of which share an
end-point. Observe that since each edge in a matching has a
distinct representative in any vertex cover, the size of a largest
matching in \(G\) is a \emph{lower bound} on the vertex cover
number of \(G\). So when input graphs have matchings of size
\(\Omega(n)\)---which is, in a real sense, \emph{the most common
  case} by far (see, e.g.;~\cite[Theorem
7.14]{Bollobas2001})---\FPT algorithms of the kind described in
the previous paragraph take \(\Omega(c^{n})\) time for some
constant \(c\), and one cannot hope to improve this to the form
\(\Oh(c^{o(n)})\) unless ETH fails. Put differently: consider the
standard parameterization of \VC, and let \(MM\) denote the size
of a maximum matching---the \emph{matching number}---of the input
graph. Note that we can find \(MM\) in polynomial
time~\cite{Edmonds1965}. When \(k<MM\) the answer is trivially
\NO, and thus such instances are uninteresting, and when
\(k\geq{}MM\), \FPT algorithms for this parameterization of the
problem are impractical for those (most common) instances for
which \(MM=\Omega(n)\).

Such considerations led to the following alternative
parameterization of \VC, where the parameter is the ``excess''
above the matching number:

\defparproblem{\AGVC(\agvc)}%
{A graph \(G\) %, a maximum matching \(M\) of \(G\) of size
               %\(MM\),
  and \(k_{\mu}\in{}\mathbb{N}\).}%
{\(k_{\mu}\)}%
{Does \(G\) have a vertex cover of size at most \(MM+k_{\mu}\)?}

The parameterized complexity of \agvc was % open for a long time
% till it was
settled by Razgon and O'Sullivan~\cite{RazgonOSullivan2008} in
2008; they showed that the problem is \FPT and can be solved in
\(\OhStar(15^{k_{\mu}})\) time. A sequence of faster \FPT
algorithms followed: In 2011, Raman et
al.~\cite{RamanRamanujanSaurabh2011} improved the running time to
\(\OhStar(9^{k_{\mu}})\), and then Cygan et
al.~\cite{CyganPilipczukPilipczukWojtaszczyk2011,CyganPilipczukPilipczukWojtaszczyk2013}
improved it further to \(\OhStar(4^{k_{\mu}})\). In 2012,
Narayanaswamy et al.~\cite{NarayanaswamyRamanRamanujanSaurabh2012}
developed a faster algorithm which solved \agvc in time
\(\OhStar(2.618^{k_{\mu}})\).  Lokshtanov et
al.~\cite{LokshtanovNarayanaswamyRamanRamanujanSaurabh2014}
improved on this to obtain an algorithm with a running time of
\(\OhStar(2.3146^{k_{\mu}})\).  This is currently the fastest \FPT
algorithm for \AGVC.

The algorithms of Narayanaswamy et al. and Lokshtanov et al. in
fact solve a ``stricter'' parameterization of \VC. Let \(LP\)
denote the minimum value of a solution to the \emph{linear
  programming relaxation} of the standard LP formulation of \VC
(See \autoref{sec:preliminaries} for definitions.). % It follows
% directly from definitions that
Then \(LP\) is a lower bound on the vertex cover number of the
graph. Narayanaswamy et al. introduced the following
parameterization\footnote{To be precise, Narayanaswamy et al. used
  the value \(\lceil{}LP\rceil\) instead of just \(LP\) in their
  definition of \VCAL, but this makes no essential difference.}
of \VC, ``above'' \(LP\):

\defparproblem{\VCAL(\vcal)}%
{A graph \(G\) %, the minimum value \(LP\) of a solution to the LP
  % relaxation of the standard LP formulation of \VC for \(G\),
  and
  \(k_{\lambda}\in{}\mathbb{N}\).}%
{\(k_{\lambda}\)}%
{Does \(G\) have a vertex cover of size at most
  \(LP+k_{\lambda}\)?}

The two algorithms solve \vcal in times
\(\OhStar(2.618^{k_{\lambda}})\) and
\(\OhStar(2.3146^{k_{\lambda}})\), respectively. Since the
inequality \(LP\geq{}MM\) holds for every graph we get that \vcal
is a \emph{stricter} parameterization of \VC, in the sense that
these algorithms for \vcal directly imply algorithms which solve
\agvc in times \(\OhStar(2.618^{k_{\mu}})\) and
\(\OhStar(2.3146^{k_{\mu}})\), respectively. To see this, consider
an instance \((G,k_{\mu})\) of \agvc where \(MM\) is the matching
number and \(VC_{opt}\) is the (unknown) vertex cover number of
the input graph \(G\), and let \(x=MM+k_{\mu}\). The question is
then whether \(VC_{opt}\leq{}x\). To resolve this, find the value
\(LP\) for the graph \(G\) (in polynomial time) and set
\(k_{\lambda}=x-LP\). Now we have that
\(VC_{opt}\leq{}x\iff{}VC_{opt}\leq{}LP+k_{\lambda}\), and we can
check if the latter inequality holds---that is, we can solve
\vcal---in \(\OhStar(2.3146^{k_{\lambda}})\) time using the
algorithm of Lokshtanov et al. Now
\(LP\geq{}MM\implies(x-LP)\leq(x-MM)\implies{k_{\lambda}}\leq{}k_{\mu}\),
and so this algorithm runs in \(\OhStar(2.618^{k_{\mu}})\) time as
well.

This leads us naturally to the next question: can we push this
further? Is there an even stricter lower bound for \VC, such that
\VC is still fixed-parameter tractable when parameterized above
this bound? To start with, it is not clear that a stricter lower
bound even exists for \VC: what could such a bound possibly look
like? It turns out that we can indeed derive such a bound; we are
then left with the task of resolving tractability above this
stricter bound.

\medskip\noindent\textbf{Our Problem.} Motivated by an observation
of Lov\'{a}sz and Plummer~\cite{LovaszPlummerBook2009} we
show---see \autoref{lem:lower_bound}---that the quantity
\((2LP-MM)\) is a lower bound on the vertex cover number of a
graph, and since \(LP\geq{}MM\) we get that
\((2LP-MM)\geq{}LP\). This motivates the following
parameterization of \VC:

\defparproblem{\VCALP(\vcalp)}%
{A graph \(G\) %, the minimum value \(LP\) of a solution to the LP
  % relaxation of the standard LP formulation of \VC for \(G\), the
  % size \(MM\) of a maximum matching of \(G\),
  and \(\hat{k}\in{}\mathbb{N}\).}%
{\(\hat{k}\)}%
{Does \(G\) have a vertex cover of size at most
  \((2LP-MM)+\hat{k}\)?}

Since \((2LP-MM)\geq{}LP\) we get, following similar arguments as
described above, that \VCALP is a stricter parameterization than
\VCAL. Further, \((k_{\lambda}-\hat{k})\) could be as large as
\((LP-MM)\) and---to the best of our knowledge---this difference
cannot be expressed as a function of \(k_{\lambda}\) alone for the
purpose of solving \VC. These facts justify our choice of
parameter: \VCALP is indeed a stricter parameterization than both
\AGVC and \VCAL, and its tractability does not follow directly
from known results.

\noindent\textbf{Our Results.} The main result of this work is that \VC is 
fixed-parameter tractable even when parameterized above this
stricter lower bound:
\begin{theorem}\label{thm:main}
  \VCALP is fixed-parameter tractable and can be solved in
  \(\OhStar(3^{\hat{k}})\) time.
\end{theorem}
By the discussions above, this directly implies similar \FPT
algorithms for the two weaker parameterizations:
\begin{corollary}\label{cor:main}
  \VCAL can be solved in \(\OhStar(3^{k_{\lambda}})\) time, and
  \AGVC can be solved in \(\OhStar(3^{k_{\mu}})\) time.
\end{corollary}

\medskip\noindent\textbf{Our Methods.} We now sketch the main
ideas behind our \FPT algorithm for \VCALP; the details are in
\autoref{sec:algorithm}. Let \(k=(2LP-MM)+\hat{k}\) denote the
``budget'', which is the maximum size of a vertex cover in whose
existence we are interested; we want to find if there is a vertex
cover of size at most \(k\).  At its core, our algorithm is a
simple branching algorithm which (i) manages to drop the
\emph{measure} \(\hat{k}\) by at least \(1\) on each branch, and
(ii) has a worst-case branching factor of \(3\). To achieve this
kind of branching we need to find, in polynomial time,
constant-sized structures in the graph, on which we can branch
profitably. It turns out that none of the ``local'' branching
strategies traditionally used to attack \VC---from the simplest
``pick an edge and branch on its end-points'' to more involved
ones based on complicated structures (e.g., those of Chen et
al.~\cite{ChenKanjXia2010})---serve our purpose. All of these
branching strategies give us a drop in \(k\) in each
branch---because we pick one or more vertices into the
solution---but give us \emph{no control} over how \(LP\) and
\(MM\) change.
 
So we turn to the more recent ideas of Narayanaswamy et
al.~\cite{NarayanaswamyRamanRamanujanSaurabh2012} who solve a
similar problem: they find a way to \emph{preprocess} the input
graph using reduction rules in such a way that branching on small
structures in the resulting graph would make \emph{their measure}
\((k-LP)\) drop on each branch. We find that their reduction rules
do not increase our measure, and hence that we can safely apply
these rules to obtain a graph where we can pick up to two vertices
in each branch and still have control over how \(LP\) changes.
The \emph{branching rules} of Narayanaswamy et al., however, are
not of use to us: these rules help control the drop in \(LP\), but
they provide no control over how \(MM\) changes. Note that for our
measure \(\hat{k}\) to drop we need, roughly speaking, (i) a good
drop in \(k\), (ii) a small drop in \(LP\), and (iii) a \emph{good
  drop} in \(MM\). None of the branching strategies of
Narayanaswamy et al.  or Lokshtanov et al. (or others in the
literature which we tried) help with this.

To get past this point we look at the classical
\emph{Gallai-Edmonds decomposition} of the reduced graph, which
can be computed in polynomial
time~\cite{Gallai1963,Gallai1964,Edmonds1965,LovaszPlummerBook2009}. We
prove that by carefully choosing edges to branch based on this
decomposition, we can ensure that both \(LP\) and \(MM\) change in
a way which gives us a net drop in the measure \(\hat{k}\). The
key ingredient and the most novel aspect of our algorithm---as
compared to existing algorithms for \VC---is the way in which we
exploit the Gallai-Edmonds decomposition to find small
structures---edges and vertices---on which we can branch
profitably. While this part is almost trivial to implement, most
of the technical effort in the paper has gone into proving that
our choices are correct. See \autoref{algorithm_outline} for an
outline which highlights the new parts, and \autoref{algorithm} on
page~\pageref{algorithm} for the complete algorithm.

\begin{algorithm}[t]
  \caption{An outline of the algorithm for \VCALP.}\label{algorithm_outline}
\begin{algorithmic}[1]
\Function{VCAL-P}{$(G,\hat{k})$}
  \State Exhaustively apply the three reduction rules of Narayanaswamy et al. to
  \((G,\hat{k})\). 
  \State Let \((G,\hat{k})\) denote the resulting graph on which no rule applies.
  \If{\((G,\hat{k})\) is a trivial instance}
    \State \textbf{return} \texttt{True} or \texttt{False} as appropriate.
  \EndIf
  %\Statex
  \State Compute the Gallai-Edmonds decomposition \(V(G)=O\uplus{}I\uplus{}P\) of \(G\).
  %\Statex
  \If {\(G[I\cup{}P]\) contains at least one edge \(\{u,v\}\)}
    \State Branch on the edge \(\{u,v\}\). \(\hat{k}\) drops by
    \(1\) on each branch.
  %\ElsIf {\(I\neq\emptyset\)} 
  %  \State Choose a suitable edge \(\{i,o\}\;;\;i\in{}I,o\in{}O\) 
  %  \State Branch on the edge \(\{i,o\}\):
  %  \State \hspace{\algorithmicindent}When we pick \(i\) into the
  %  solution in one branch, \(\hat{k}\) drops by \(1\).
  %  \State \hspace{\algorithmicindent}When we pick \(o\) into the
  %  solution in the other branch, \(\hat{k}\) \emph{may not}
  %  drop. We find a suitable edge in \(G'=(G\setminus{}o)\) and
  %  branch on its end-points to make \(\hat{k}\) drop.
  \Else\Comment{Now \(P=\emptyset\).}
    \State Branch on a vertex \(u\in{}O=V(G)\) and two of its neighbours \(v,w\in{O}\): 
    \State \hspace{\algorithmicindent}When we pick both of \(v,w\) into the
    solution in one branch, \(\hat{k}\) drops by \(1\).
    \State \hspace{\algorithmicindent}When we pick \(u\) into the
    solution in the other branch, \(\hat{k}\) \emph{may not}
    drop. We find a suitable edge in  \(G'=(G\setminus{}u)\) and
    branch on its end-points to make \(\hat{k}\) drop.
  \EndIf
\EndFunction
\end{algorithmic}
\end{algorithm}

% \todo[inline,disable]{Aside: Can we juggle the values involved in
%   \measure to get a faster algorithm for the ``basic''
%   parameterized \VC? If \(MM\) is too large we can say no. I have
%   a hunch that \(MM\) cannot be too small compared to \(OPT\)
%   (say, less than \(OPT/2\). For intermediate values of \(MM\),
%   what does the drop in our measure translate to in terms of
%   \(OPT\)?}
 
% \medskip\noindent\textbf{Related Work.} The general idea of
% above/below guarantee problems and a bit of history, other such
% problems (survey?).

% \medskip\noindent\textbf{Organization of the Rest of the Paper.}

%%% Local Variables: 
%%% mode: latex
%%% TeX-master: "main"
%%% End: 

%% file: preliminaries.tex
\section{Preliminaries}\label{sec:preliminaries}
We use \(\uplus\) to denote the disjoint union of sets. All our
graphs are undirected and simple. \(V(G)\) and \(E(G)\) denote,
respectively, the vertex and edge sets of a graph \(G\).  \(G[X]\)
is the subgraph of \(G\) \emph{induced} by a vertex subset
\(X\subseteq{V(G)}\):
\(G[X]=(X,F)\;;\;F=\{\{v,w\}\in{E(G)}\;;\;v,w\in{X}\}\).
\(MM(G)\) is the matching number of graph \(G\), and \(OPT(G)\) is
the vertex cover number of \(G\). A matching \(M\) in graph \(G\)
\emph{saturates} each vertex which is an end-point of an edge in
\(M\), and \emph{exposes} every other vertex in \(G\). \(M\) is a
\emph{perfect matching} if it saturates all of \(V(G)\). \(M\) is
a \emph{near-perfect matching} if it saturates all but one vertex
of \(V(G)\). Graph \(G\) is said to be \emph{factor-critical} if
for each \(v\in{V(G)}\) the induced subgraph
\(G[(V(G)\setminus\{v\})]\) has a perfect matching.

For \(X\subseteq{V(G)}\), \(N(X)\) is the set of neighbours of
\(X\) which are not in \(X\):
\(N(X)=\{v\in{(V(G)\setminus{X})}\;;\;\exists{w}\in{X}\;:\;\{v,w\}\in{E(G)}\}\).
\(X\subseteq{V(G)}\) is an \emph{independent set} in graph \(G\)
if no edge in \(G\) has both its end-points in \(X\).  The
\emph{surplus} of an independent set \(X\subseteq{V(G)}\) is
\(\surplus{X}=(|N(X)|-|X|)\).  The \emph{surplus of a graph
  \(G\)}, \(\surplus{G}\), is the minimum surplus over all
independent sets in \(G\).  Graph \(G\) is a \emph{bipartite
  graph} if \(V(G)\) can be partitioned as \(V(G)=X\uplus{Y}\)
such that every edge in \(G\) has exactly one end point in each of
the sets \(X,Y\). Hall's Theorem tells us that a bipartite graph
\(G=((X\uplus{Y}),E)\) contains a matching which saturates all
vertices of the set \(X\) if and only if
\(\;\forall{S}\subseteq{X}\;:\;|N(S)|\geq|S|\). K\"{o}nig's
Theorem tells us that for a bipartite graph \(G\),
\(OPT(G)=MM(G)\).

The \emph{linear programming (LP) relaxation} of the standard LP
formulation for \VC for a graph \(G\) (the \emph{relaxed \VC LP
  for \(G\)} for short), denoted \(LPVC(G)\), is:
\begin{alignat*}{2}\label{vclp}
    \text{minimize }   & \sum_{v\in{V(G)}}x_{v}\ \\
    \text{subject to } & x_u + x_v \geq 1\  &,\ & \{u,v\} \in E\\
                       & 0\leq{x_v}\leq1\ &,\ & v\in{V(G)}
\end{alignat*}
A \emph{feasible solution} to this LP is an assignment of values
to the variables \(x_{v}\;;\;v\in{V(G)}\) which satisfies all the
conditions in the LP, and an \emph{optimum solution} is a feasible
solution which minimizes the value of the objective function
\(\sum_{v\in{V(G)}}x_{v}\). We use \(w(x)\) to denote the value
(of the objective function) of a feasible solution \(x\) to
\(LPVC(G)\), and \(LP(G)\) to denote the value of an optimum
solution to \(LPVC(G)\). \(OPT(G)\) and \(MM(G)\) are then the
values of optimum solutions to the \emph{integer} programs
corresponding to \(LPVC(G)\) and to its LP \emph{dual},
respectively~\cite{BourjollyPulleyblank1989}. It follows that for
any graph \(G\), \(MM(G)\leq{LP(G)}\leq{OPT(G)}\).  Our stronger
lower bound for \(OPT(G)\) is motivated by a similar bound due to
Lov\'{a}sz and
Plummer~\cite[Theorem~6.3.3]{LovaszPlummerBook2009}:

\begin{lemma}\label{lem:lower_bound}
  For any graph \(G\), \(OPT(G)\geq(2LP(G)-MM(G))\).
\end{lemma}
\begin{proof} 
  Let \(S\) be a smallest vertex cover of graph \(G\), and let
  \(H=((S\uplus(V(G)\setminus{S})),F)\) be the bipartite subgraph
  of \(G\) with
  \(F=\{u,v\}\in{E(G)}\;;\;u\in{S},v\in(V(G)\setminus{S})\). That
  is, the vertex set of \(H\) is \(V(G)\) with the bipartition
  \((S,(V(G)\setminus{S}))\), and the edge set of \(H\) consists
  of exactly those edges of \(G\) which have one end-point in
  \(S\) and the other in \((V(G)\setminus{S})\).  Let \(T\) be a
  smallest vertex cover of graph \(H\). Then \(|T|=OPT(H)=MM(H)\),
  where the second equality follows from K\"{o}nig's
  Theorem. Consider the following assignment \(y\) of values to
  variables \(y_{v}\;;\;v\in{V(G)}\):
  \[
  y_{v}=
  \begin{cases}
    1 & \text{if } v\in(S\cap{T}) \\
    \frac{1}{2} & \text{if }v\in((S\cup{T})\setminus(S\cap{T}))\\
    0 & \text{otherwise}.
  \end{cases}
  \]
  
  Observe that \(0\leq{y_{v}}\leq1\) for each \(v\in{V(G)}\). We
  claim that \(y\) is a feasible solution to \(LPVC(G)\). Indeed,
  since \(S\) is a vertex cover of \(G\), every edge
  \(\{u,v\}\in{E(G)}\) must have at least one end-point in
  \(S\). If \(\{u,v\}\cap(S\cap{T})\neq\emptyset\) then \(y\)
  assigns the value \(1\) to at least one of \(y_{u},y_{v}\), and
  so we have that \(y_{u}+y_{v}\geq1\). Otherwise, if
  \(\{u,v\}\subseteq(S\setminus(S\cap{T}))\) then \(y\) assigns
  the value \(\frac{1}{2}\) to both of \(y_{u},y_{v}\), and so we
  have that \(y_{u}+y_{v}=1\). In the only remaining case, exactly
  one of \(\{u,v\}\) is in \(S\), and the other vertex is in
  \((V(G)\setminus{S})\). Without loss of generality, suppose
  \(\{u,v\}\cap{S}=\{u\},v\in(V(G)\setminus{S})\). Then
  \(\{u,v\}\in{E(H)}\) and \(u\notin{T}\), hence we get---since
  \(T\) is a vertex cover of \(H\)---that \(v\in{T}\). Thus
  \(\{u,v\}\subseteq((S\cup{T})\setminus(S\cap{T}))\), and so
  \(y\) assigns the value \(\frac{1}{2}\) to both of
  \(y_{u},y_{v}\) and we have that \(y_{u}+y_{v}=1\). Thus the
  assignment \(y\) satisfies all the conditions in the
  \(LPVC(G)\), and hence is a feasible solution to
  \(LPVC(G)\). Thus \(w(y)\geq{LP(G)}\). 

  Observe now that
  \[w(y) = \frac{|S|+|T|}{2} = \frac{OPT(G)+OPT(H)}{2} =
  \frac{OPT(G)+MM(H)}{2} \leq \frac{OPT(G)+MM(G)}{2},\]
  where the first equality follows from the way we defined \(y\),
  and the inequality follows from the observation that the
  matching number of the subgraph \(H\) of \(G\) cannot be
  \emph{larger} than that of \(G\) itself. Putting these together
  we get that \(LP(G)\leq{w(y)}\leq\frac{OPT(G)+MM(G)}{2}\), which
  in turn gives the bound in the lemma.
\end{proof}

For any graph \(G\) there exists an \emph{optimum} solution to
\(LPVC(G)\) in which
\(x_{v}\in\{0,\frac{1}{2},1\}\;;\;v\in{V(G)}\)~\cite{NemhauserTrotter1974}. Such
an optimum solution is called a \emph{half-integral solution} to
\(LPVC(G)\), and we can find such a solution in polynomial
time~\cite{NemhauserTrotter1975}. Whenever we refer to an optimum
solution to \(LPVC(G)\) in the rest of the paper, we mean a
half-integral solution. Given a half-integral solution \(x\) to
\(LPVC(G)\), we define \(V_{i}^{x}=\{v\in{V(G)}\;;\;x_{v}=i\}\)
for each \(i\in\{0,\frac{1}{2},1\}\). For any optimal
half-integral solution \(x\) we have that
\(N(V^{x}_{0})=V^{x}_{1}\). Given a graph \(G\) as input we can,
in polynomial time, compute an optimum half-integral solution
\(x\) to \(LPVC(G)\) such that for the induced subgraph
\(H=G[V_{1/2}^{x}]\), setting all variables to the value
\(\frac{1}{2}\) is the \emph{unique} optimum solution to
\(LPVC(H)\)~\cite{NemhauserTrotter1975}. Graphs which satisfy the
latter property must have positive surplus, and conversely:

\begin{lemma}\textup{\textbf{\cite{Pulleyblank1979}}}\label{lem:all-halves_positive_surplus}
  For any graph \(G\), all-\(\frac{1}{2}\) is the unique optimum
  solution to \(LPVC(G)\) if and only if \(\surplus{G}>0\).
\end{lemma}

In fact, the surplus of a graph \(G\) is a lower bound on the
number of vertices that we can delete from \(G\), and have a
\emph{guaranteed} drop of \emph{exactly} \(\frac{1}{2}\) \emph{per
  deleted vertex} in \(LP(G)\):
\begin{lemma}\label{lem:surplus_lp_drop}
  Let \(G\) be a graph with \(\surplus{G}\geq{s}\). Then deleting
  any subset of \(s\) vertices from \(G\) results in a graph
  \(G'\) such that \(LP(G')=LP(G)-\frac{s}{2}\).
\end{lemma}
\begin{proof}
  The proof is by induction on \(s\). Let \(n=|V(G)|\). The case
  \(s=0\) is trivially true. Suppose \(s=1\). Then by
  \autoref{lem:all-halves_positive_surplus} all-\(\frac{1}{2}\) is the
  unique optimum solution to \(LPVC(G)\), and so
  \(LP(G)=\frac{n}{2}\). Let \(v\in{V(G)}\), and let
  \(G'=G[V(G)\setminus\{v\}]\). Then \(|V(G'|=(n-1)\). Since
  all-\(\frac{1}{2}\) is a \emph{feasible} solution to
  \(LPVC(G')\), we get that
  \(LP(G')\leq\frac{(n-1)}{2}=LP(G)-\frac{1}{2}\). If possible,
  let \(x'\) be an optimum solution for \(LPVC(G')\) such that
  \(w(x')<\frac{(n-1)}{2}\). % Without loss of generality, we may
  % assume that \(x'\) is a half-integral solution, and so
  From the half-integrality property of relaxed \VC LP
  formulations we get that \(w(x')\leq(\frac{n}{2}-1)\). Now we
  can assign the value \(1\) to vertex \(v\) and the values \(x'\)
  to the remaining vertices of graph \(G\), to get a solution
  \(x\) such that \(w(x)=\frac{n}{2}=LP(G)\). Thus \(x\) is an
  \emph{optimum} solution to \(LPVC(G)\) which is \emph{not}
  all-\(\frac{1}{2}\), a contradiction. So we get that
  \(LP(G')=LP(G)-\frac{1}{2}\), proving the case \(s=1\).

  For the induction step, let \(s\geq2\). Then by
  \autoref{lem:all-halves_positive_surplus} all-\(\frac{1}{2}\) is
  the unique optimum solution to \(LPVC(G)\), and so
  \(LP(G)=\frac{n}{2}\). Let \(v\) be an arbitrary vertex in
  \(G\), and let \(G'=G[V(G)\setminus\{v\}]\). Since deleting a
  single vertex from \(G\) cannot cause the surplus of \(G\) to
  drop by more than \(1\), we get that
  \(\surplus{G'}\geq(s-1)\geq1\). So from
  \autoref{lem:all-halves_positive_surplus} we get that
  \(LP(G')=\frac{(n-1)}{2}\).  Applying the induction hypothesis
  to \(G'\) and \(s-1\), we get that deleting any subset of
  \(s-1\) vertices from \(G'\) results in a graph \(G''\) such
  that
  \(LP(G'') = LP(G')-\frac{(s-1)}{2} =
  \frac{(n-1)}{2}-\frac{(s-1)}{2} = \frac{(n-s)}{2}\). This
  completes the induction step.
\end{proof}

There is a matching between the vertex sets which get the values
\(0\) and \(1\) in an optimal half-integral solution to the LP.

\begin{lemma}\label{lem:bimatch}
  Let \(G\) be a graph, and let \(x\) be an optimal half-integral
  solution to \(LPVC(G)\). Let
  \(H=((V^{x}_{1}\uplus{V^{x}_{0}}),F)\) be the bipartite subgraph
  of \(G\) where
  \(F=\{u,v\}\in{E(G)}\;;\;u\in{V^{x}_{1}},v\in{V^{x}_{0}}\). Then
  there exists a (maximum) matching of \(H\) which saturates all
  of $V^x_1$.
\end{lemma}
\begin{proof} It is enough to show that for every
  \(X\subseteq{V^{x}_{1}}\) the inequality \(|N(X)|\geq|X|\) holds
  in the graph \(H\), and the rest will follow by Hall's Theorem.
  So let there exist some \(X\subseteq{V^{x}_{1}}\) such that
  \(|N(X)|<|X|\) in \(H\). Now consider a solution \(x^{*}\) to
  \(LPVC(G)\) in which all the vertices of \(G\) have same values
  as in \(x\), \emph{except} that all vertices in \(X\cup{N(X)}\)
  get the value \(\frac{1}{2}\). It is not difficult to verify
  that \(x^{*}\) is a feasible solution to \(LPVC(G)\). Now
  \(w(x^{*})=w(x)+\frac{|N(X)|-|X|}{2}<w(x)\), which is a
  contradiction since we assumed that \(x\) is a solution to
  \(LPVC(G)\) with the minimum value. The lemma follows.
\end{proof}

We make critical use of the classical Gallai-Edmonds decomposition
of graphs. 
\begin{definition}[Gallai-Edmonds decomposition]\label{def:gallai_edmonds_decomposition}
  The Gallai-Edmonds decomposition of a graph \(G\) is a partition
  of its vertex set \(V(G)\) as \(V(G)=O\uplus{I}\uplus{P}\)
  where:
  \begin{itemize}
  \item
    \(O = \{v\in{V(G)}\;;\; \text{ some \emph{maximum} matching of
    } G \text{ leaves } v \text{ exposed}\}\)
  \item \(I = N(O)\)
  \item \(P = V(G)\setminus(I\cup{O})\)
  \end{itemize}
\end{definition}

We now list a few of the many useful properties of Gallai-Edmonds
decompositions.
\begin{theorem}\textup{\textbf{\cite{LovaszPlummerBook2009,Gallai1964,Gallai1963,Edmonds1965}}}\label{thm:gallai_edmonds_properties}
  The Gallai-Edmonds decomposition of a graph \(G\) is unique, and
  can be computed in polynomial time in the size of \(G\). Let
  \(V(G)=O\uplus{I}\uplus{P}\) be the Gallai-Edmonds decomposition
  of \(G\). Then the following hold:
\begin{enumerate}
\item Every component of the induced subgraph \(G[O]\) is
  factor-critical.
\item A matching \(M\) in graph \(G\) is a \emph{maximum} matching
  of \(G\) if and only if:
\begin{enumerate}
\item For each connected component \(H\) of the induced subgraph
  \(G[O]\), the edge set \(M\cap{E(H)}\) forms a \emph{near
    perfect} matching of \(H\);
\item For each vertex \(i\in{I}\) there exists some vertex
  \(i\in{O}\) such that \(\{i,o\}\in{M}\), and;
\item The edge set \(M\cap{E(G[P])}\) forms a \emph{perfect}
  matching of the induced subgraph \(G[P]\).
\end {enumerate}
\item In particular: Any maximum matching \(M\) of \(G\) is a
  disjoint union of (i) a perfect matching of \(G[P]\), (ii)
  near-perfect matchings of each component of \(G[O]\), and (iii)
  an edge from each vertex in \(I\) to a distinct component of
  \(G[O]\).
\item The Stability Lemma: For a vertex \(v\in{V(G)}\) let
  \(G-v={G[V(G)\setminus{v}]}\). Let \(O(H),I(H),P(H)\) denote the
  three parts in the Gallai-Edmonds decomposition of a graph
  \(H\).
  \begin{itemize}
  \item Let \(v\in{O}\). Then
    \(O(G-v)\subseteq(O\setminus\{v\})\), \(I(G-v)\subseteq{I}\),
    and \(P(G-v)\supseteq{P}\).
  \item Let \(v\in{I}\). Then \(O(G-v)=O\),
    \(I(G-v)=(I\setminus\{v\})\), and \(P(G-v)=P\).
  \item Let \(v\in{P}\). Then \(O(G-v)\supseteq{O}\),
    \(I(G-v)\supseteq{I}\), and
    \(P(G-v)\subseteq(P\setminus\{v\})\).
  \end{itemize}
\end {enumerate}
\end{theorem}
\begin{corollary}\label{cor:gallai_edmonds_more_properties}
  Let \(V(G)=O\uplus{I}\uplus{P}\) be the Gallai-Edmonds
  decomposition of graph \(G\). Then the following hold:
  \begin{enumerate}
  \item If \(I\cup{P}\) is an independent set in \(G\), then
    \(P=\emptyset\).
  \item Let \(v,w\in{O}\) be two vertices which are part of the
    same connected component of the induced subgraph \(G[O]\), and
    let \(G'=G[(V(G)\setminus\{v,w\})]\). Then
    \(MM(G')\leq(MM(G)-1)\).
  \end{enumerate}
\end{corollary}
\begin{proof}
  We prove each statement.
  \begin{enumerate}
  \item From the assumption, \(G[P]\) contains no edges. By
    \autoref{thm:gallai_edmonds_properties}, \(G[P]\) has a
    perfect matching. Both of these can hold simultaneously only
    when \(P=\emptyset\).
  \item Let \(C\) be the connected component of \(G[O]\) which
    contains both \(v\) and \(w\). From part (3) of
    \autoref{thm:gallai_edmonds_properties} we get that any
    maximum matching of graph \(G\) which \emph{exposes} vertex
    \(v\) contains a \emph{perfect} matching of the subgraph
    \(C''=C[(V(C)\setminus\{v\})]\). Therefore, every maximum
    matching of \(G\) which survives in
    \(G''=G[(V(G)\setminus\{v\})]\) contains a perfect matching of
    \(C''\).  It follows that if we delete \(w\in{V(C'')}\) as
    well from \(G''\) to get \(G'\), then the matching number
    reduces by at least one, since no perfect matching of \(C''\)
    can survive the deletion of vertex \(w\in{V(C'')}\).
  \end{enumerate}
\end{proof}

When the set \(I\cup{P}\) is independent in a graph of surplus at
least \(2\), we get more properties for the set \(O\):

\begin{lemma}\label{lem:egbipart}
  Let \(G\) be a graph with \(\surplus{G}\geq2\), and let
  \(V(G)=O\uplus{I}\uplus{P}\) be the Gallai-Edmonds decomposition
  of graph \(G\).  If \(I\cup{P}\) is an independent set in \(G\),
  then: 

  \begin{enumerate}
  \item There is at least one vertex \(o\in{O}\) which has at
    least two neighbours in the set \(O\).
  \item Let \(v\in{O}\) be a neighbour of some vertex \(o\in{O}\)% ,
    % and let \(X=(\{v\}\cup(N(\{v\})\cap(O\setminus\{o\})))\)
    .  Let \(G'= G[(V(G)\setminus\{o\})]\), and let
    \(V(G')=O'\uplus{I'}\uplus{P'}\) be the Gallai-Edmonds
    decomposition of graph \(G'\).  Then the induced subgraph
    \(G'[P']\) contains at least one edge% ; indeed,
    % \(X\subseteq{P'}\)
    .
  
  \end{enumerate}
\end{lemma}
\begin{proof} 
  We prove each statement.

  \begin{enumerate}
  \item Each component of \(G[O]\) is factor critical
    (\autoref{thm:gallai_edmonds_properties}), and so has an odd
    number of vertices.  So to prove this part it is enough to
    show that there is at least one component in \(G[O]\) which
    has at least two (and hence, at least three) vertices.
    Suppose to the contrary that each component in \(G[O]\) has
    exactly one vertex. Then the set \(O\) is an independent set
    in graph \(G\).  Since the set \(I\cup{P}\) is independent by
    assumption, we get from
    \autoref{cor:gallai_edmonds_more_properties} that
    \(P=\emptyset\). Thus \(G\) is a \emph{bipartite} graph with
    \(O\) and \(I\) being two parts of the bipartition. In
    particular, \(N(O)\subseteq{I}\) and \(N(I)\subseteq{O}\) both
    hold. But since \(\surplus{G}\geq2\) this implies that both
    \(|I|\geq(|O|+2)\) and \(|O|\geq(|I|+2)\) hold simultaneously,
    which cannot happen. The claim follows.
  \item Observe that by the definition of the set \(O\), we have
    that \(MM(G')=MM(G)\). If \(G'[P']\) contains no edge, then
    the fact that \(G[P']\) has a perfect matching
    (\autoref{thm:gallai_edmonds_properties}) implies that
    \(P'=\emptyset\).  Together with the Stability Lemma (see
    \autoref{thm:gallai_edmonds_properties}) this implies that
    \(v\in{O'}\) in graph \(G'\). Then by the definition of the
    set \(O'\) there is a maximum matching \(MM'\) of \(G'\) which
    exposes the vertex \(v\).  But then the matching \(MM'\)
    together with the edge \(\{o,v\}\) forms a matching in graph
    \(G\) of size \(MM(G')+1=MM(G)+1\), a contradiction. Therefore
    \(G'[P']\) must contain at least one edge% , and the vertex
    % \(v\) must belong to the set \(P'\) in graph \(G'\)
    .
    % Since there are no edges between the sets \(O'\) and \(P'\) by
    % definition, and since \(v\in{P'}\) in graph \(G'\), no vertex
    % in the set \(X\subseteq{N(v)}\) can be in the set \(O'\) in
    % graph \(G'\). By the Stability Lemma, no vertex in the set
    % \(X\subseteq{O}\) can be in the set \(I'\) either. It follows
    % that \(X\subseteq{P'}\).
  \end{enumerate}
\end{proof}

%% file: algorithm.tex
\section{The Algorithm}\label{sec:algorithm}
In this section we describe our algorithm which solves \VCALP in
\(\OhStar(3^{\hat{k}})\) time. We start with an overview of the
algorithm. We then state the reduction and branching rules which
we use, and prove their correctness. We conclude the section by
proving that the algorithm correctly solves \VCALP and that it
runs within the stated running time.

\subsection{Overview}\label{sec:overview}
% We first set up some notation which we use in this section. We use
% \(LP(G)\) to denote the minimum value of a solution to the linear
% programming relaxation of \VC for a graph \(G\), \(MM(G)\) to
% denote the matching number of \(G\), and \(OPT(G)\) to denote the
% vertex cover number of \(G\). 
The \VCALP problem can be stated as:

\defparproblem{\VCALP(\vcalp)}%
{A graph \(G\) %, the numbers \(LP(G)\) and \(MM(G)\),
  and \(\hat{k}\in{}\mathbb{N}\).}%
{\(\hat{k}\)}%
{Let \(k=(2LP(G)-MM(G))+\hat{k}\). Is \(OPT(G)\leq{}k\)?}

At its core the algorithm is a simple branching algorithm. Indeed,
we employ only two basic branching steps which are both
``natural'' for \VC: We either branch on the two end-points of an
edge, or on a vertex \(v\) and a pair \(u,w\) of its neighbours.
We set \(\hat{k}\) % \(\hat{k}=(k+MM(G)-2LP(G))\)
as the \emph{measure} for analyzing the running time of the
algorithm.  We ensure that this measure decreases by \(1\) in each
branch of a branching step in the algorithm. Since
\(OPT(G)\geq(2LP(G)-MM(G))\) for every graph \(G\)---See
\autoref{lem:lower_bound}---and \(k\geq{}OPT(G)\) holds---by
definition---for a \YES instance, it follows that
\(\hat{k}=(k+MM(G)-2LP(G))\) is never negative for a \YES
instance. Hence we can safely terminate the branching at depth
\(\hat{k}\). Since the worst-case branching factor is three, we
get that the algorithm solves \vcalp in time
\(\OhStar(3^{\hat{k}})\).

\input{pseudocode.tex}
% The \textbf{while} in algorithm \ref{euclid} ends in line
% \ref{euclidendwhile}, so \algref{euclid}{euclidendwhile} is the
% line we seek.

% \todo[inline]{Add a three-line outline of the algorithm as
%   pseudocode in the Introduction. }

The algorithm---see \autoref{algorithm} on
page~\pageref{algorithm}---modifies the input graph \(G\) in
various ways as it runs. In our description of the algorithm we
will sometimes, for the sake of brevity, slightly abuse the
notation and use \(G\) to refer to the ``current'' graph at each
point in the algorithm. Since the intended meaning will be clear
from the context, this should not cause any confusion.

The crux of the algorithm is the manner in which it finds in \(G\)
a small structure---such as an edge---on which to branch such that
the measure \(\hat{k}\) drops on each branch. Clearly, picking an
arbitrary edge---say---and branching on its end-points will not
suffice: this will certainly reduce the budget \(k\) by one on
each branch, but we will have no guarantees on how the values
\(MM(G)\) and \(LP(G)\) change, and so we cannot be sure that
\(\hat{k}\) drops. Instead, we apply a two-pronged strategy to
find a small structure in the graph which gives us enough control
over how the values \(k,MM(G),LP(G)\) change, in such a way as to
ensure that \(\hat{k}=(k+MM(G)-2LP(G))\) drops in each branch.

First, we employ a set of three reduction rules which give us
control over how \(LP(G)\) changes when we pick a vertex into the
solution and delete it from \(G\). These rules have the following
nice properties:
\begin{itemize}
\item Each rule is sound---see \autoref{sec:redrules}---and can be
  applied in polynomial time;
\item Applying any of these rules does not increase the measure
  \(\hat{k}\), and;
\item If none of these rules applies to \(G\), then we can delete
  up to two vertices from \(G\) with an assured drop of \(0.5\)
  \emph{per deleted vertex} in the value of \(LP(G)\).
  % deleting any
  % one vertex from \(G\) results in a graph \(G'\) such that
  % \(LP(G')=LP(G)-0.5\), and deleting any \emph{two} vertices from
  % \(G\) results in a graph \(G'\) such that \(LP(G')=LP(G)-1\).
\end{itemize}
In this we follow the approach of Narayanaswamy et
al.~\cite{NarayanaswamyRamanRamanujanSaurabh2012} who came up with
this strategy to find a small structure on which to branch so that
their measure---which was \((k-LP)\)---would reduce on each
branch. Indeed, we \emph{reuse} three of their reduction rules,
after proving that these rules do not increase our measure
\(\hat{k}\).

While these rules give us control over how \(LP(G)\) changes when
we delete a vertex, they do not give us such control over
\(MM(G)\). That is: let \(G\) be a graph to which none of these
rules applies and let \(v\) be a vertex of \(G\). Picking \(v\)
into a solution and deleting it from \(G\) to get \(G'\) would (i)
reduce \(k\) by \(1\), and (ii) make \(LP(G')=LP(G)-0.5\), but
will \emph{not} result in a drop in \(\hat{k}\) \emph{if it so
  happens} that the matching number of \(G'\) is the same as that
of \(G\). Thus for our approach to succeed we need to be able to
consistently find, in a graph reduced with respect to the
reduction rules, a small structure---say, an edge---on which to
branch, such that deleting a small set of vertices from this
structure would reduce the matching number of the graph by at
least \(1\).

The search for such a structure led us to the second and novel
part of our approach, namely the use of the classical
\emph{Gallai-Edmonds decomposition} of graphs
(\autoref{def:gallai_edmonds_decomposition}). We found that by
first reducing the input graph with respect to the three reduction
rules and then carefully choosing edges to branch based on the
Gallai-Edmonds decomposition of the resulting graph, we could make
the measure \(\hat{k}\) drop on every branch. We describe the
reduction and branching rules in the next two subsections.

\input{reduction_rules.tex}
\input{branching_rules.tex}

\input{analysis.tex}
%%% Local Variables: 
%%% mode: latex
%%% TeX-master: "main"
%%% End: 

%% file: pseudocode.tex
\begin{algorithm}
  \caption{The \(\OhStar(3^{\hat{k}})\) \FPT algorithm for
    \VCALP.}\label{algorithm}
\begin{algorithmic}[1]
\Function{VCAL-P}{$(G,\hat{k},reduce)$}\Comment{Initially invoked with \(reduce==True\).}
  \State \(k\gets (2LP(G)-MM(G))+\hat{k}\)
  \If{\(reduce==True\)}\Comment{Do the preprocessing only if \(reduce\) is \texttt{True}.} 
  \While{At least one of Reduction Rules 1, 2, and 3 applies to \((G,k)\)} 
    \State Apply the \emph{first} such rule to \((G,k)\), to get \((G',k')\). 
    \State \(G\gets{}G',k\gets{}k'\)
  \EndWhile
  \Statex
   %\State Exhaustively apply Reduction Rule 1 to \((G,k)\) to get \((G',k')\).
   %\State Exhaustively apply Reduction Rule 2 to \((G',k')\) to get \((G'',k'')\).
   %\State Exhaustively apply Reduction Rule 3 to \((G'',k'')\) to get \((G''',k''')\).
   %\State \(G\gets{}G''',k\gets{}k'''\)\Comment{Renaming for brevity}
   \State \(G\gets{}G',k\gets{}k'\)\Comment{Renaming for brevity}
   \State \(\hat{k}\gets k+MM(G)-2LP(G)\)\Comment{Preprocessing ends.}
  \EndIf 
  \Statex 
  \If{\(\hat{k}<0\)}\Comment{Check for the base cases.}
    \State \textbf{return} \texttt{False}
  \ElsIf{\(G\) has no vertices}
    \State \textbf{return} \texttt{True}
  \EndIf
  \Statex
  \State Compute the Gallai-Edmonds decomposition \(V(G)=O\uplus{}I\uplus{}P\) of \(G\).
  \Statex
  \If {\(G[I\cup{}P]\) contains at least one edge
    \(\{u,v\}\)}\Comment{Branch on edge \(\{u,v\}\).}
    \State \(G_{1}\gets{}(G\setminus{}u),G_{2}\gets{}(G\setminus{}v)\)    
    \State \(k\gets{}(k-1)\)
    \State \(\hat{k}_{1}=(k+MM(G_{1})-2LP(G_{1})),\hat{k}_{2}=(k+MM(G_{2})-2LP(G_{2}))\)
    \State \textbf{return} (VCAL-P(\(G_{1},\hat{k}_{1},True\)) \(\vee\) VCAL-P(\(G_{2},\hat{k}_{2},True\)))
  \Statex
  %\ElsIf {\(I\neq\emptyset\)}\Comment{Since \(I\cup{}P\) is independent, \(P=\emptyset\).} 
  %  \State Choose an edge \(\{i,o\}\;;\;i\in{}I,o\in{}O\) where
  %  the component of \(o\) in \(G[O]\) contains at least one other
  %  vertex.\Comment{Such an edge must exist.}
  %  \State \(G_{1}\gets{}(G\setminus{}i),G_{2}\gets{}(G\setminus{}o)\)    
  %  \State \(k\gets{}(k-1)\)
  %  \State \(\hat{k}_{1}=(k+MM(G_{1})-2LP(G_{1})),\hat{k}_{2}=(k+MM(G_{2})-2LP(G_{2}))\)
  %  \State \textbf{return} (VCAL-P(\(G_{1},\hat{k}_{1},True\)) \(\vee\) VCAL-P(\(G_{2},\hat{k}_{2},False\)))
  %\Statex
  \Else\Comment{Since \(I\cup{}P\) is independent, \(P=\emptyset\).}
    \Statex
    \State Choose a vertex \(u\in{O}\) and two of its neighbours \(v,w\in{O}\). 
    \Statex \Comment{Such vertices must exist.}
    \State \(G_{1}\gets{}(G\setminus{}u)\)\Comment{Pick \(u\) in the solution.}
    \State \(k\gets{}(k-1)\)
    \State \(\hat{k}_{1}=(k+MM(G_{1})-2LP(G_{1}))\),
    \If
    {VCAL-P(\(G_{1},\hat{k}_{1},False\))\(==\)\texttt{True}}\Comment{No reduction rule will be applied now.}
      \State \textbf{return} \texttt{True}
     \Else\Comment{Pick both \(v\) and \(w\) in the solution.}
       \State \(G_{2}\gets{}(G\setminus{}\{v,w\})\)    
       \State \(k\gets{}(k-1)\)\Comment{Effectively reducing \(k\) by two.}
       \State \(\hat{k}_{2}=(k+MM(G_{2})-2LP(G_{2}))\)
       \State \textbf{return} VCAL-P(\(G_{2},\hat{k}_{2},True\))
    \EndIf
  \EndIf
\EndFunction
\end{algorithmic}
\end{algorithm}
%\todo{We may have to reapply starting from rule 1 after applying each of rules 2 and 3.}
%%% Local Variables: 
%%% mode: latex
%%% TeX-master: "main"
%%% End: 

%% file: reduction_rules.tex
\subsection{The Reduction Rules}\label{sec:redrules}
Given an instance \((G,\hat{k})\) of \VCALP our algorithm first
computes the number \(k=(2LP(G)-MM(G)+\hat{k})\) so that \((G,k)\)
is the equivalent instance of (classical) \VC% ---see line~2 of
% \autoref{algorithm}
.  Each of our three reduction rules takes an instance
\((G=(V,E),k)\) of \VC, runs in polynomial time, and outputs an
instance \((G'=(V',E'),k')\) of \VC. We say that a reduction rule
is \emph{sound} if it always outputs an \emph{equivalent}
instance. That is, if it is always the case that \(G\) has a
vertex cover of size at most \(k\) if and only if \(G'\) has a
vertex cover of size at most \(k'\). We say that a reduction rule
is \emph{safe} if it never increases the measure \(\hat{k}\). That
is, if it is always the case that
\(k'+MM(G')-2LP(G')\leq{}k+MM(G)-2LP(G)\).

In the algorithm we apply these reduction rules
\emph{exhaustively}, in the order they are presented. That is, we
take the input instance \((G,k)\) and apply the first among
\autoref{red:NT_reduction}, \autoref{red:edgeinneighbour}, and
\autoref{red:struction} which \emph{applies}---see definitions
below---to \((G,k)\) to obtain a modified instance \((G',k')\). We
now set \(G\gets{G'},k\gets{k'}\) and repeat this procedure, till
none of the rules applies to the instance \((G,k)\). We say that
such an instance is \emph{reduced} with respect to our reduction
rules.  The point of these reduction rules is that they help us
push the \emph{surplus} of the input graph to at least \(2\) in
polynomial time, while not increasing the measure \(\hat{k}\).

% That is, if
% \autoref{red:NT_reduction} applies to the input instance \((G,k)\)
% then we apply the rule to get an instance \((G',k')\) to which the
% rule no longer applies. Then if \autoref{red:edgeinneighbour} applies
% to \((G',k')\) then we apply the rule to get an instance
% \((G_{1},k_{1})\). If \autoref{red:edgeinneighbour} applies
% to \((G_{1},k_{1})\) then we apply the rule to get an instance
% \((G_{2},k_{2})\), and so on till we get an instance \((G'',k'')\)
% to which the rule no longer applies. We repeat this procedure with
% \autoref{red:struction} and . 

\begin{reduction}\label{red:NT_reduction} Compute an optimal solution
  \(x\) to LPVC(\(G\)) such that all-\(\frac{1}{2}\) is the unique
  optimum solution to LPVC(\(G[V^{x}_{1/2}]\)). Set
  \(G'=G[V^{x}_{1/2}],k'=k-|V^{x}_{1}|\).
\end{reduction}
\noindent \autoref{red:NT_reduction} \emph{applies} to \((G,k)\)
if and only if all-\(\frac{1}{2}\) is \emph{not} the unique
solution to LPVC(\(G\)).

\begin{reduction}\label{red:edgeinneighbour} 
  % If \autoref{red:NT_reduction} does \emph{not} apply to
  % \((G,k)\), then do the following. 
  If there is an independent set $Z\subseteq{}V(G)$ such that
  \(\surplus{Z}=1\) and \(N(Z)\) is \emph{not} an independent set
  in \(G\), then set \(G'=G\setminus{(Z\cup{}N(Z))},k'=k-|N(Z)|\).
\end{reduction}
\noindent \autoref{red:edgeinneighbour} \emph{applies} to
\((G,k)\) if and only if (i) \autoref{red:NT_reduction} does
\emph{not} apply, and (ii) \(G\) has a vertex subset \(Z\) with
the properties stated in \autoref{red:edgeinneighbour}.

\begin{reduction} \label{red:struction} % If neither of
  % \autoref{red:NT_reduction} and \autoref{red:edgeinneighbour}
  % applies to \((G,k)\), then do the following.
  If there is an independent set $Z\subseteq{}V(G)$ such that
  \(\surplus{Z}=1\) and \(N(Z)\) \emph{is} an independent set in
  \(G\), then remove $Z$ from \(G\) and \emph{identify} the
  vertices of $N(Z)$---that is, delete all of \(N(Z)\), add a new
  vertex \(z\), and make \(z\) adjacent to all vertices of the set
  \((N(N(Z))\setminus{Z})\)---to get \(G'\), and set \(k'=k-|Z|\).
\end{reduction}
\noindent \autoref{red:struction} \emph{applies} to \((G,k)\) if
and only if (i) neither of the previous rules applies, and (ii)
\(G\) has a vertex subset \(Z\) with the properties stated in
\autoref{red:struction}.

All our reduction rules are due to Narayanaswamy et
al.~\cite[Preprocessing Rules 1 and
2]{NarayanaswamyRamanRamanujanSaurabh2012}. The soundness of these
rules and their running-time bounds also follow directly from
their work. We need to argue, however, that these rules are safe
for our measure \(\hat{k}\).

\begin{lemma}\label{lem:rules_sound_fast}\textup{\textbf{\cite{NarayanaswamyRamanRamanujanSaurabh2012}}}
  All the three reduction rules are sound, and each can be applied
  in time polynomial in the size of the input \((G,k)\).
\end{lemma}

It remains to show that none of these reduction rules increases
the measure \(\hat{k}=k+MM(G)-2LP(G)\). That is, let
\((G=(V,E),k)\) be an instance of \VC to which one of these rules
applies, and let \((G'=(V',E'),k')\) be the instance obtained by
applying the rule to \((G,k)\). Then we have to show, for each
rule, that \(\hat{k}'=k'+MM(G')-2LP(G')\leq{}\hat{k}\) holds. We
establish this by examining how each rule changes the values
\(k\), \(MM\) and \(LP\). In each case, let \(x\) be an optimum
half-integral solution to \(LPVC(G)\) such that
all-\(\frac{1}{2}\) is the unique optimum solution to
\(LPVC(G[V^{x}_{1/2}])\), and let \(LP(G)\) be the (optimum) value
of this solution. Also, let \(x'\) be an optimum half-integral
solution to \(LPVC(G')\), and let \(LP(G')\) be the value of this
solution.

\begin{lemma}\label{lem_rule_one_safe}
\autoref{red:NT_reduction} is safe.
\end{lemma}
\begin{proof}
  From the definition of the rule we get that
  \(k'=k-|V^{x}_{1}|\). Now since \(x'\equiv\frac{1}{2}\) is the
  unique optimum solution to \(LPVC(G')\), and
  \(G'=G[V^{x}_{1/2}]\), we get that \(LP(G')=LP(G)-|V^{x}_{1}|\),
  and hence that \(2LP(G')=2LP(G)-2|V^{x}_{1}|\). From
  \autoref{lem:bimatch} and the construction of the graph \(G'\)
  we get that
  \(MM(G)\geq% {}MM(G[V^{x}_{1/2}])+|V^{x}_{1}|=
  MM(G')+|V^{x}_{1}|\),
  and hence that \(MM(G')\leq{}MM(G)-|V^{x}_{1}|\). Putting these
  together, we get that \(\hat{k}'\leq\hat{k}\).
\end{proof}
% A graph on which \autoref{red:NT_reduction} does not apply has the
% following useful property.
% \begin{lemma}\label{lem:all_halves_surplus_positive}\textup{\textbf{\cite{NarayanaswamyRamanRamanujanSaurabh2012}}}
%   If \autoref{red:NT_reduction} does not apply to \((G,k)\) then
%   \(\surplus{G}\geq{}1\).
% \end{lemma}
\begin{lemma}\label{lem_rule_two_safe}
\autoref{red:edgeinneighbour} is safe.
\end{lemma}
\begin{proof}
  From the definition of the rule we get that \(k'=k-|N(Z)|\). We
  bound the other two quantities.
  \begin{claim}
    \(LP(G')\geq{}LP(G)-|N(Z)|+\frac{1}{2}\).
  \end{claim}
  \begin{proof}[Proof of the claim.]
    Since \autoref{red:NT_reduction} did not apply to \((G,k)\) we
    know that all-\(\frac{1}{2}\) is the unique optimum solution
    to \(LPVC(G)\). It follows from this and the construction of
    graph \(G'\) that
    \(LP(G')=\sum_{u\in{}V'}x'(u)=LP(G)-\frac{1}{2}(|Z|+|N(Z)|)+\frac{1}{2}(|V^{x'}_{1}|-|V^{x'}_{0}|)\).
    Adding and subtracting \(\frac{1}{2}(|N(Z)|)\), we get
    \(LP(G')=LP(G)-|N(Z)|+\frac{1}{2}(|N(Z)|-|Z|)+\frac{1}{2}(|V^{x'}_{1}|-|V^{x'}_{0}|)
    =LP(G)-|N(Z)|+\frac{1}{2}(|N(Z)|+|V^{x'}_{1}|)-\frac{1}{2}(|Z|+|V^{x'}_{0}|)\).
    Now since \(V(G')=(V(G)\setminus{}(Z\cup{}N(Z)))\) and
    \(V^{x'}_{0}\subseteq{}V(G')\), we get that
    \(Z\cup{}V^{x'}_{0}\) is an independent set in \(G\), and that
    \(N(Z\cup{}V^{x'}_0)=N(Z)\cup{}V^{x'}_{1}\) in \(G\) (Recall
    that \(N(V^{x}_{0})=V^{x}_{1}\) for any half-integral optimal
    solution \(x\).).  Since
    \(\surplus{G}\geq{}1\)---\autoref{lem:all-halves_positive_surplus}---we
    get that in \(G\),
    \(|N(Z\cup{}V^{x'}_{0})|-|Z\cup{}V^{x'}_{0}|\geq{}1\), which
    gives \((|N(Z)|+|V^{x'}_{1}|)-(|Z|+|V^{x'}_{0}|)\geq{}1\), and
    so we get that
    \(\frac{1}{2}(|N(Z)|+|V^{x'}_{1}|)-\frac{1}{2}(|Z|+|V^{x'}_{0}|)\geq{}\frac{1}{2}\).
    Substituting this in the equation for \(LP(G')\), we get that
    \(LP(G')\geq{}LP(G)-|N(Z)|+\frac{1}{2}\).
  \end{proof}
  Now we bound the drop in \(MM(G)\).
  \begin{claim}
    \(MM(G')\leq{}MM(G)-|Z|\).
  \end{claim}
  \begin{proof}[Proof of the claim.]
    Consider the bipartite graph \(\hat{G}\) obtained from the
    induced subgraph \(G[Z\cup{}N(Z)]\) of \(G\) by deleting every
    edge which has both endpoints in \(N(Z)\). Observe that since
    \(\surplus{Z}=1\) in \(G\), we get that
    $|N(X)|\geq|X|+1\;\forall{}X\subseteq{}Z$, both in \(G\) and
    in \(\hat{G}\). Hence by Hall's Theorem we get that
    \(\hat{G}\) contains a matching saturating \(Z\), and hence
    that \(MM(\hat{G})=|Z|\). But from the construction of graph
    \(G'\) we get that
    \(MM(G)\geq{}MM(G')+MM(\hat{G})\). Substituting for
    \(MM(\hat{G})\), we get that \(MM(G')\leq{}MM(G)-|Z|\).
  \end{proof}
  
  Putting all these together, we get that
  \(\hat{k}'=k'+MM(G')-2LP(G')\leq{}(k-|N(Z)|)+(MM(G)-|Z|)-2(LP(G)-|N(Z)|+\frac{1}{2})=k+MM(G)-2LP(G)+(|N(Z)|-|Z|-1)=k+MM(G)-2LP(G)=\hat{k}\),
  where the last-but-one equality follows from the fact that
  \(\surplus{Z}=1\). Thus we have that \(\hat{k}'\leq{}\hat{k}\).
\end{proof}

\begin{lemma}\label{lem:rule_three_safe}
\autoref{red:struction} is safe.
\end{lemma}
\begin{proof}
  From the definition of the rule we get that \(k'=k-|Z|\). We
  bound the other quantities. Let \(z\) be the vertex in \(G'\)
  which results from identifying the vertices of \(N(Z)\) as
  stipulated by the reduction rule.
  \begin{claim}
    \(LP(G')\geq LP(G)-|Z|\).
  \end{claim}
  \begin{proof}[Proof of the claim.]
    Suppose not. Then we get---from the half-integrality property
    of the relaxed LP for \VC---that
    \(LP(G')\leq{}LP(G)-|Z|-\frac{1}{2}\). We show that this
    inequality leads to a contradiction. We consider three cases
    based on the value of \(x'(z)\), which must be one of
    \(\{0,\frac{1}{2},1\}\). Recall that \(\surplus{Z}=1\) and so
    \(|N(Z)|=|Z|+1\).
    \begin{description}
    \item[Case 1:] \(x'(z)=1\). Consider a function
      \(x'':V\to\{0,\frac{1}{2},1\}\) defined as follows. For
      every vertex \(v\) in $V'\setminus\{z\}$, \(x''\) retains
      the value assigned by \(x'\); that is, \(x''(v)=x'(v)\). For
      every vertex \(v\) in the set \(N(Z)\), set \(x''(v)=1\) and
      for every vertex \(v\) in the set \(Z\), \(x''(v)=0\). It is
      not difficult to check that \(x''\) is a feasible solution
      to the relaxed \VC LP for \(G\).  But now the value of this
      solution,
      \(w(x'')=LP(G')-x'(z)+|N(Z)|=LP(G')-1+(|Z|+1)\leq{}LP(G)-\frac{1}{2}\). Thus
      \(x''\) is a feasible solution for \(G\) whose value is less
      than that of the \emph{optimal} solution, a contradiction.

    \item[Case 2:] \(x'(z)=0\). Now consider the following
      function \(x'':V\to\{0,\frac{1}{2},1\}\). For every vertex
      \(v\) in $V'\setminus\{z\}$, \(x''\) retains the value
      assigned by \(x'\): \(x''(v)=x'(v)\). For every vertex
      \(v\in{}Z\) set \(x''(v)=1\) and for every vertex
      \(v\in{}N(Z)\), set \(x''(v)=0\). It is again not difficult
      to check that \(x''\) is a feasible solution to the relaxed
      \VC LP for \(G\). And the value of this solution, \(w(x'')=
      LP(G')+|Z|\leq{}LP(G)-\frac{1}{2}\), again a contradiction.

    \item[Case 3:] \(x'(z)=\frac{1}{2}\). Consider again a
      function \(x'':V\to\{0,\frac{1}{2},1\}\), defined as
      follows. For every vertex \(v\) in \(V'\setminus\{z\}\),
      \(x''\)---once again---retains the value assigned by \(x'\):
      \(x''(v)=x'(v)\). For every vertex \(v\in(Z\cup{}N(Z))\),
      set \(x''(v)=\frac{1}{2}\). This is once again easily
      verified to be a feasible solution to the relaxed \VC LP for
      \(G\), and its value is
      \(w(x'')=LP(G')-x'(z)+\frac{1}{2}(|Z|+|N(Z)|)=LP(G')-\frac{1}{2}+\frac{1}{2}(|Z|+|Z|+1)\leq{}LP(G)-\frac{1}{2}\),
      once again a contradiction.
    \end{description}
    Thus our contrary assumption leads to a contradiction in all
    possible cases, and this proves the claim.
  \end{proof}
    
  Now we bound the drop in \(MM(G)\).
  \begin{claim}
    \(MM(G')\leq{}MM(G)-|Z|\).
  \end{claim}
  \begin{proof}
    For an arbitrary vertex \(u\in{}N(Z)\), let \(G_{u}\) be the
    bipartite subgraph \(G[Z\cup{}(N(Z)\setminus{}\{u\})]\), which
    is an induced subgraph of \(G\).  Since \(\surplus{Z}=1\) in
    \(G\) and only one vertex from \(N(Z)\) is missing in
    \(G_{u}\), we get that
    \(|N(X)|\geq|X|\;\forall{}X\subseteq{}Z\) holds in
    \(G_{u}\). Hence by Hall's Theorem we get that \(G_{u}\)
    contains a matching saturating \(Z\), and hence that
    \(MM(G_{u})=|Z|\).
      
    Now consider a maximum matching \(MM'\) of \(G'\). Starting
    with \(MM'\) we can construct a matching \(MM''\) of size
    \(|MM'|\) in the original graph \(G\) which saturates at most
    one vertex of \(N(Z)\) and none of \(Z\), as follows: There is
    at most one edge in \(MM'\) which saturates the vertex
    \(z\). If there is \emph{no} edge in \(MM'\) which saturates
    the vertex \(z\), then we set \(MM''=MM'\). It is not
    difficult to see that \(MM''\) saturates no vertex in
    \(Z\cup{}N(Z)\). If there is an edge \(\{z,v\}\in{}MM'\), then
    we pick an arbitrary vertex \(u\in{}N(Z)\) such that
    \(\{u,v\}\) is an edge in \(G\)---such a vertex must exist
    since the edge \(\{z,v\}\) exists in \(G'\). We set
    \(MM''=(MM'\setminus\{\{z,v\}\})\cup\{\{u,v\}\}\). It is not
    difficult to see that \(MM''\) is a matching in \(G\) which
    saturates exactly one vertex---\(u\in{}N(Z)\)---in
    \(Z\cup{}N(Z)\).
      
    If the matching \(MM''\), constructed as above, does not
    saturate any vertex of \(N(Z)\), then we choose \(u\) to be an
    arbitrary vertex of \(N(Z)\). If \(MM''\) does saturate a
    vertex of \(N(Z)\), then we set \(u\) to be that vertex. In
    either case, the union of \(MM''\) and any maximum matching of
    the induced bipartite subgraph \(G_{u}\) is itself a matching
    of \(G\), of size \(MM(G_{u})+MM(G')=|Z|+MM(G')\). It follows
    that \(MM(G)\geq|Z|+MM(G')\), which implies
    \(MM(G')\leq{}MM(G)-|Z|\).
    \end{proof}
    Putting all these together, we get that
    \(\hat{k}'=k'+MM(G')-2LP(G')\leq{}(k-|Z|)+(MM(G)-|Z|)-2(LP(G)-|Z|)=k+MM(G)-2LP(G)=\hat{k}\). Thus
    we have that \(\hat{k}'\leq{}\hat{k}\).
\end{proof}

Observe that if \(\surplus{G}=1\) then at least one of
\autoref{red:edgeinneighbour} and \autoref{red:struction}
necessarily applies. From this and
\autoref{lem:all-halves_positive_surplus} we get the following
useful property of graphs on which none of these rules applies.

\begin{lemma}\label{lem:reduced_graph_surplus_two}
  None of the three reduction rules applies to an instance
  \((G,k)\) of \VC if and only if \(\surplus{G}\geq{}2\).
\end{lemma}

Summarizing the results of this section, we have: 
\begin{lemma}\label{lem:result_of_reduction_rules}
  Given an instance \((G,\hat{k})\) of \VCALP we can, in time
  polynomial in the size of the input, compute an instance
  \((G',\hat{k}')\) such that:
  \begin{enumerate}
  \item The two instances are equivalent: \(G\) has a vertex cover
    of size at most \((2LP(G)-MM(G))+\hat{k}\) if and only if
    \(G'\) has a vertex cover of size at most
    \((2LP(G')-MM(G'))+\hat{k}'\);
  \item \(\hat{k}'\leq\hat{k}\), and \(\surplus{G'}\geq{}2\).
  \end{enumerate}
\end{lemma}
%%% Local Variables: 
%%% mode: latex
%%% TeX-master: "main"
%%% End: 

%% file: branching_rules.tex
\subsection{The Branching Rules}
Let \((G=(V,E),k)\) be the instance of \VC obtained after
exhaustively applying the reduction rules to the input instance,
and let \(\hat{k}=k+MM(G)-2LP(G)\). Then \(\surplus{G}\geq{}2\)% ;
% in particular, each vertex of \(G\) has degree at least \(3\)
. If \(\hat{k}<0\) then the instance \((G,\hat{k})\) of \vcalp is
trivially a \NO instance (See \autoref{sec:overview}), and we
return \NO and stop.  In the remaining case \(\hat{k}\geq0\), and
we apply one of two branching rules to the instance to reduce the
measure \(\hat{k}\). Each of our branching rules takes an instance
\((G,\hat{k})\) of \vcalp, runs in polynomial time, and outputs
either two or three instances
\((G_{1},\hat{k}_{1}),(G_{2},\hat{k}_{2})[,(G_{3},\hat{k}_{3})]\)
of \vcalp. The algorithm then recurses on each of these instances
in turn. We say that a branching rule is \emph{sound} if the
following holds: \((G,\hat{k})\) is a \YES instance of \vcalp if
and only if \emph{at least one} of the (two or three) instances
output by the rule is a YES instance of \vcalp. We now present the
branching rules, prove that they are sound, and show that the
measure \(\hat{k}\) drops by at least \(1\) on each branch of each
rule.

Before starting with the branching, we compute the
\emph{Gallai-Edmonds decomposition} of the graph \(G=(V,E)\). This
can be done in polynomial time~\cite{LovaszPlummerBook2009} and
yields a partition of the vertex set \(V\) into three parts
\(V=O\uplus{}I\uplus{}P\), one or more of which may be
empty---see~\autoref{def:gallai_edmonds_decomposition}. We then
branch on edges in the graph which are carefully chosen with
respect to this partition. \autoref{bra:IP_branch} applies to the
graph \(G\) if and only
if % it is not part of a trivial instance, and
the induced subgraph (\(G[I\cup{}P]\)) contains at least one edge.
 
\begin{branching}\label{bra:IP_branch}
  % If the induced subgraph \(G[I\cup{}P]\) contains at least one
  % edge, then 
  Branch on the two end-points of an edge \(\{u,v\}\) in the
  induced subgraph (\(G[I\cup{}P]\)). % That is: take vertex \(u\)
  % into the vertex cover on one branch, and vertex \(v\) into the
  % vertex cover on the other branch.
  More precisely, the two
  branches % modify the graph, \(k\), and \(\hat{k}\)
  generate two instances as follows:
  \begin{description} 
    \item[Branch 1:] \(G_{1}\gets{}(G\setminus{}\{u\}),k_{1}\gets(k-1),\hat{k_{1}}\gets{}k_{1}+MM(G_{1})-2LP(G_{1})\)
    \item[Branch 2:] \(G_{2}\gets{}(G\setminus{}\{v\}),k_{2}\gets(k-1),\hat{k_{2}}\gets{}k_{2}+MM(G_{2})-2LP(G_{2})\)
  \end{description}
  This rule outputs the two instances \((G_{1},\hat{k}_{1})\) and \((G_{2},\hat{k}_{2})\).
\end{branching}

%%%%%%%%%%%%%%%%%%%%%%%%%%%%%%%%%%%%%%%%%%%%%%%%%%%%%%%%%%%%%%%
% \begin{branching}\label{bra:IO_branch}
%   If \autoref{bra:IP_branch} does \emph{not} apply to graph \(G\)
%   and the set \(I\) is nonempty, then do the following. Pick an
%   edge \(\{i,o\}\) of \(G\) where \(i\in{}I,o\in{}O\) such that
%   vertex $o$ has atleast one neighbour in $O$. Construct the graph
%   \(G'=G\setminus{}\{o\}\) and compute its Gallai-Edmonds
%   decomposition \(O'\uplus{}I'\uplus{}P'\). Pick an edge
%   \(\{u,v\}\) in \(G'[P']\). The three branches of this branching
%   rule modify the graph, \(k\), and \(\hat{k}\) as follows:
%   \begin{description} 
%     \item[Branch 1:] \(G_{1}\gets{}(G\setminus{}\{i\}),k_{1}\gets(k-1),\hat{k_{1}}\gets{}k_{1}+MM(G_{1})-2LP(G_{1})\)
%     \item[Branch 2:] \(G_{2}\gets{}(G'\setminus{}\{u\}),k_{2}\gets(k-2),\hat{k_{2}}\gets{}k_{2}+MM(G_{2})-2LP(G_{2})\)
%     \item[Branch 3:] \(G_{3}\gets{}(G'\setminus{}\{v\}),k_{3}\gets(k-2),\hat{k_{3}}\gets{}k_{3}+MM(G_{3})-2LP(G_{3})\)
%   \end{description}
%   This rule outputs the three instances \((G_{1},\hat{k}_{1})\),
%   \((G_{2},\hat{k}_{2})\), and \((G_{3},\hat{k}_{3})\).
% \end{branching}
% \noindent \autoref{bra:IO_branch} \emph{applies} to graph \(G\) if
% and only if it is not part of a trivial instance,
% \autoref{bra:IP_branch} does \emph{not} apply to \(G\), and the
% set \(I\) in the Gallai-Edmonds partition of \(G\) is nonempty.
%%%%%%%%%%%%%%%%%%%%%%%%%%%%%%%%%%%%%%%%%%%%%%%%%%%%%%%%%%%%%%%

\noindent \autoref{bra:O_branch} applies to graph \(G\) if and
only if % it is not part of a trivial instance, and
\autoref{bra:IP_branch} does not apply to \(G\).
\begin{branching}\label{bra:O_branch}
  % If the induced subgraph \(G[I\cup{P}]\) contains no edge,
  % % \autoref{bra:IP_branch} does not apply to graph \(G\)
  % then do the following.
  Pick a vertex \(u\in{O}\) of \(G\) which
  has at least two neighbours \(v,w\in{O}\)% , and let \(v,w\) be any two
  % neighbours of \(u\) which are in \(O\)
  . Construct the graph \(G'=G\setminus{}\{u\}\) and compute its
  Gallai-Edmonds decomposition \(O'\uplus{}I'\uplus{}P'\). Pick an
  edge \(\{x,y\}\) in \(G'[P']\). The three branches of this
  branching rule % modify the graph, \(k\), and \(\hat{k}\)
  generates three instances as follows:
  \begin{description} 
    \item[Branch 1:] \(G_{1}\gets{}(G\setminus{}\{v,w\}),k_{1}\gets(k-2),\hat{k_{1}}\gets{}k_{1}+MM(G_{1})-2LP(G_{1})\)
    \item[Branch 2:] \(G_{2}\gets{}(G'\setminus{}\{x\}),k_{2}\gets(k-2),\hat{k_{2}}\gets{}k_{2}+MM(G_{2})-2LP(G_{2})\)
    \item[Branch 3:] \(G_{3}\gets{}(G'\setminus{}\{y\}),k_{3}\gets(k-2),\hat{k_{3}}\gets{}k_{3}+MM(G_{3})-2LP(G_{3})\)
  \end{description}
  This rule outputs the three instances \((G_{1},\hat{k}_{1})\),
  \((G_{2},\hat{k}_{2})\), and \((G_{3},\hat{k}_{3})\).
\end{branching}

Note that in the pseudocode of \autoref{algorithm} we have not,
for the sake of brevity, written out the three branches of the
second rule. Instead, we have used a boolean switch (\(reduce\))
as a third argument to the procedure to simulate this behaviour.

\begin{lemma}\label{lem:branching_sound}
  Both branching rules are sound, and each can be applied
  in time polynomial in the size of the input \((G,k)\).
\end{lemma}
\begin{proof}
  We first prove that the rules can be applied in polynomial
  time. Recall that we can compute the Gallai-Edmonds
  decomposition of graph \(G\) in polynomial
  time~\cite{LovaszPlummerBook2009}.
  \begin{description}
  \item[\autoref{bra:IP_branch}:] Follows more or less directly
    from the definition of the branching rule. 
    % \item[\autoref{bra:IO_branch}:] It is not difficult to see
    %   that the applicability of this rule can be checked in
    %   polynomial time. If the rule does apply, then we get from
    %   Lemma~\ref{lem:egbipart} and Lemma~\ref{lem:delshift} that
    %   edges \(\{i,o\}\) and \(\{u,v\}\) of the kind required in
    %   the rule will necessarily exist. It is not difficult to see
    %   that given these edges, the rest of the rule can be applied
    %   in polynomial time.
  \item[\autoref{bra:O_branch}:] It is not difficult to see that
    the applicability of this rule can be checked in polynomial
    time. If the rule does apply, then the set \(P\) is empty and
    the set \(I\) is an independent set in \(G\)
    (\autoref{cor:gallai_edmonds_more_properties}).  By
    \autoref{lem:egbipart}, there is at least one vertex in the
    set \(O\) which has at least two neighbours in \(O\). Hence
    vertices \(u,v,w\) of the kind required in the rule
    necessarily exist in \(G\). From \autoref{lem:egbipart} we
    also get that an edge \(\{x,y\}\) of the required kind also
    exists. Given the existence of these, it is not difficult to
    see that the rule can be applied in polynomial time.
  \end{description}
  We now show that the rules are sound. Recall that a branching
  rule is sound if the following holds: any instance
  \((G,\hat{k})\) of \vcalp is a \YES instance if and only if at
  least one of the instances obtained by applying the rule to
  \((G,\hat{k})\) is a \YES instance.
  \begin{description}
  \item[\autoref{bra:IP_branch}:] The soundness of this rule is
    not difficult to see since it is a straightforward branching
    on the two end-points of an edge. Nevertheless, we include a
    full argument for the sake of completeness.  So let
    \((G,\hat{k})\) be a \YES instance of \vcalp, and let
    \((G_{1},\hat{k_{1}}),(G_{2},\hat{k_{2}})\) be the two
    instances obtained by applying \autoref{bra:IP_branch} to
    \((G,\hat{k})\). Since \((G,\hat{k})\) is a \YES instance, the
    graph \(G\) has a vertex cover, say \(S\), of size at most
    \(k=2LP(G)-MM(G)+\hat{k}\). Then
    \((S\cap\{u,v\})\neq{}\emptyset\). Suppose \(u\in{}S\). Then
    \(S_{1}=S\setminus{}\{u\})\) is a vertex cover of the graph
    \(G_{1}=G\setminus{}\{u\}\), of size at most
    \(k_{1}=(k-1)\). It follows that for
    \(\hat{k_{1}}=k_{1}+MM(G_{1})-2LP(G_{1})\),
    \((G_{1},\hat{k_{1}})\) is a \YES instance of \vcalp. A
    symmetric argument gives us the \YES instance
    \((G_{2},\hat{k_{2}})\) of \vcalp for the case \(v\in{}S\).
    
    Conversely, suppose \((G_{1},\hat{k_{1}})\) is a \YES instance
    of \vcalp where
    \(G_{1}=G\setminus{}\{u\},k_{1}=(k-1),\hat{k_{1}}=k_{1}+MM(G_{1})-2LP(G_{1})\). Then
    graph \(G_{1}\) has a vertex cover, say \(S_{1}\), of size at
    most \(k_{1}=(k-1)\). \(S=S_{1}\cup\{u\}\) is then a vertex
    cover of graph \(G\) of size at most \(k\), and so
    \((G,\hat{k})\), where \(\hat{k}=k+MM(G)-2LP(G)\) is a \YES
    instance of \vcalp as well. An essentially identical argument
    works for the case when \((G_{2},\hat{k_{2}})\) is a \YES
    instance of \vcalp.

    % \item[\autoref{bra:IO_branch}:] It is not difficult to see
    %   that the applicability of this rule can be checked in
    %   polynomial time. If the rule does apply, then we get from
    %   Lemma~\ref{lem:egbipart} and Lemma~\ref{lem:delshift} that
    %   edges \(\{i,o\}\) and \(\{u,v\}\) of the kind required in
    %   the rule will necessarily exist. It is not difficult to see
    %   that given these edges, the rest of the rule can be applied
    %   in polynomial time.

  \item[\autoref{bra:O_branch}:] It is not difficult to see that
    this rule is sound, either, since it consists of exhaustive
    branching on a vertex \(u\). We include the arguments for
    completeness. So let \((G,\hat{k})\) be a \YES instance of
    \vcalp, and let
    \((G_{1},\hat{k_{1}}),(G_{2},\hat{k_{2}}),(G_{2},\hat{k_{2}})\)
    be the three instances obtained by applying
    \autoref{bra:O_branch} to \((G,\hat{k})\). Since
    \((G,\hat{k})\) is a \YES instance, the graph \(G\) has a
    vertex cover, say \(S\), of size at most
    \(k=2LP(G)-MM(G)+\hat{k}\). We consider two cases:
    \(u\in{S}\), and \(u\notin{S}\).  First consider the case
    \(u\notin{S}\). Then all the neighbours of \(u\) must be in
    \(S\). In particular, \(\{v,w\}\subseteq{S}\). It follows that
    the set \(S\setminus\{v,w\}\) is a vertex cover of the graph
    \(G_{1}=(G\setminus\{v,w\})\), of size at most
    \(k_{1}=(k-2)\). Hence we get that for
    \(\hat{k_{1}}=k_{1}+MM(G_{1})-2LP(G_{1})\),
    \((G_{1},\hat{k_{1}})\) is a \YES instance of \vcalp.
    
    Now consider the case \(u\in{S}\). Then the set
    \(S'=S\setminus\{u\}\) is a vertex cover of the graph
    \(G'=G\setminus\{u\}\), of size at most \(k-1\). Since
    \(\{x,y\}\) is an edge in the graph \(G'\), we get that
    \((S'\cap\{x,y\})\neq\emptyset\). Suppose \(x\in{}S'\). Then
    \(S_{2}=(S'\setminus\{x\})\) is a vertex cover of the graph
    \(G_{2}=G'\setminus{}\{x\}\), of size at most
    \(k_{2}=(k-2)\). It follows that for
    \(\hat{k_{2}}=k_{2}+MM(G_{2})-2LP(G_{2})\),
    \((G_{2},\hat{k_{2}})\) is a \YES instance of \vcalp. A
    symmetric argument gives us the \YES instance
    \((G_{3},\hat{k_{3}})\) of \vcalp for the case \(y\in{}S\).

    Conversely, suppose \((G_{1},\hat{k_{1}})\) is a \YES instance
    of \vcalp where
    \(G_{1}=G\setminus\{v,w\},k_{1}=(k-2),\hat{k_{1}}=k_{1}+MM(G_{1})-2LP(G_{1})\). Then
    graph \(G_{1}\) has a vertex cover, say \(S_{1}\), of size at
    most \(k_{1}=(k-2)\). \(S=S_{1}\cup\{v,w\}\) is then a vertex
    cover of graph \(G\) of size at most \(k\), and so
    \((G,\hat{k})\) where \(\hat{k}=k+MM(G)-2LP(G)\), is a \YES
    instance of \vcalp as well. Essentially identical arguments
    work for the cases when \((G_{2},\hat{k_{2}})\) or
    \((G_{3},\hat{k_{3}})\) is a \YES instance of \vcalp.
  \end{description}
\end{proof}

Our choice of vertices on which we branch ensures that the measure
drops by at least one on each branch of the algorithm.

\begin{lemma}\label{lem:measure_drop}
  Let \((G,\hat{k})\) be an input given to one of the branching
  rules, and let \((G_{i},\hat{k}_{i})\) be an instance output by
  the rule. Then \(\hat{k}_{i}\leq(\hat{k}-1)\).
\end{lemma}
\begin{proof} 
  Recall that by definition, \(\hat{k}=k+MM(G)-2LP(G)\).  We
  consider each branching rule. We reuse the notation from the
  description of each rule.
  \begin{description}
  \item[\autoref{bra:IP_branch}:] Consider \textbf{Branch 1}.
    Since \(u\in(I\cup{P})\), \emph{every} maximum matching of
    graph \(G\) saturates vertex \(u\)
    (\autoref{thm:gallai_edmonds_properties}). Hence we get that
    \(MM(G_{1})=(MM(G)-1)\). Since \(\surplus{G}\geq2\) we
    get---from \autoref{lem:surplus_lp_drop}---that
    \(LP(G_{1})=(LP(G)-\frac{1}{2})\). And since \(k_{1}=(k-1)\)
    by definition, we get that
    \(\hat{k}_{1} = k_{1}+MM(G_{1})-2LP(G_{1}) =
    (k-1)+(MM(G)-1)-2(LP(G)-\frac{1}{2}) =k+MM(G)-2LP(G)-1 =
    (\hat{k}-1)\). An essentially identical argument applied to
    the (symmetrical) \textbf{Branch 2} tells us that
    \(\hat{k}_{2}=(\hat{k}-1)\).
  \item[\autoref{bra:O_branch}:]Consider \textbf{Branch 1}. Since
    \(\{u,v,w\}\subseteq{O}\) and \(v,w\) are neighbours of vertex
    \(u\), we get that vertices \(v\) and \(w\) belong to the same
    connected component of the induced subgraph \(G[O]\) of
    \(G\). It follows
    (\autoref{cor:gallai_edmonds_more_properties}) that
    \(MM(G_{1})\leq(MM(G)-1)\). Since \(\surplus{G}\geq2\) and
    \(G_{1}=G\setminus\{v,w\}\) we get---from
    \autoref{lem:surplus_lp_drop}---that
    \(LP(G_{1})=(LP(G)-1)\). And since \(k_{1}=(k-2)\) by
    definition, we get that
    \(\hat{k}_{1} = k_{1}+MM(G_{1})-2LP(G_{1})\leq
    (k-2)+(MM(G)-1)-2(LP(G)-1) = k+MM(G)-2LP(G)-1 = (\hat{k}-1)\).
    
    Now consider \textbf{Branch 2}. Since \(u\in{O}\) we
    get---from the definition of the Gallai-Edmonds
    decomposition---that \(MM(G')=MM(G)\). Now since \(x\in{P'}\)
    we get---~\autoref{thm:gallai_edmonds_properties}---that
    \(MM(G_{2})=(MM(G')-1)=(MM(G)-1)\). Since \(\surplus{G}\geq2\)
    and \(G_{2}=G\setminus\{u,x\}\) we get---from
    \autoref{lem:surplus_lp_drop}---that
    \(LP(G_{2})=(LP(G)-1)\). And since \(k_{2}=(k-2)\) by
    definition, we get that
    \(\hat{k}_{2} = k_{2}+MM(G_{2})-2LP(G_{2}) =
    (k-2)+(MM(G)-1)-2(LP(G)-1) = k+MM(G)-2LP(G)-1 =
    (\hat{k}-1)\). An essentially identical argument applied to
    the (symmetrical) \textbf{Branch 3} tells us that
    \(\hat{k}_{3}=(\hat{k}-1)\).
  \end{description}
\end{proof}

%% file: analysis.tex
\subsection{Putting it All Together: Correctness and Running Time Analysis}
The correctness of our algorithm and the claimed bound on its
running time follow more or less directly from the above
discussion. 

\begin{proof}[Proof of \autoref{thm:main}]
  We claim that \autoref{algorithm} solves \VCALP in
  \(\OhStar(3^{\hat{k}})\) time. The correctness of the algorithm
  follows from the fact that both the reduction rules and the
  branching rules are sound---\autoref{lem:rules_sound_fast} and
  \autoref{lem:branching_sound}. This means that (i) no reduction
  rule ever converts a \YES instance into a \NO instance or
  \emph{vice versa}, and (ii) each branching rule outputs \emph{at
    least one} \YES instance when given a \YES instance as output,
  and \emph{all} \NO instances when given a \NO instance as
  input. It follows that if the algorithm outputs \YES or \NO,
  then the input instance must also have been \YES or \NO,
  respectively.

  The running time bound follows from three factors. Firstly, all
  the reduction rules are safe, and so they never increase the
  measure \(\hat{k}\)---see
  \autoref{lem_rule_one_safe},~\autoref{lem_rule_two_safe},
  \autoref{lem:rule_three_safe},
  and~\autoref{lem:result_of_reduction_rules}. Also, each
  reduction rule can be executed in polynomial time, and since
  each reduction rule reduces the number of vertices in the graph
  by at least one, they can be exhaustively applied in polynomial
  time. Secondly, each branch of each branching rule reduces the
  measure by at least \(1\)---\autoref{lem:measure_drop}---from
  which we get that each path from the root of the recursion
  tree---where the measure is the original value of
  \(\hat{k}\)---to a leaf---where the measure first becomes zero
  or less---has length at most \(\hat{k}\). Further, the largest
  branching factor is \(3\), which means that the number of nodes
  in the recursion tree is \(O(3^{\hat{k}})\). Thirdly, we know
  that the computation at each node in the recursion tree takes
  polynomial time. This follows from \autoref{lem:branching_sound}
  for the nodes where we do branching. As for the leaf nodes:
  Since we know that the measure will never be negative for a \YES
  instance (\autoref{lem:lower_bound}), since we can check for the
  applicability of each branching rule in polynomial time, and
  since each branching rule reduces the measure by at least one,
  we can solve the instance at each leaf node in polynomial time.
\end{proof}

%%% Local Variables: 
%%% mode: latex
%%% TeX-master: "main"
%%% End: 

%% file: conclusion.tex
\section{Conclusion}\label{sec:conclusion}
Motivated by an observation of Lov\'{a}sz and Plummer, we derived
the new lower bound \(2LP(G)-MM(G)\) for the vertex cover number
of a graph \(G\). This bound is at least as large as the bounds
\(MM(G)\) and \(LP(G)\) which have hitherto been used as lower
bounds for investigating above-guarantee parameterizations of \VC.
We took up the parameterization of the \VC problem above our
``higher'' lower bound \(2LP(G)-MM(G)\), which we call the
Lov\'{a}sz-Plummer lower bound for \VC. We showed that \VC remains
fixed-paramter tractable even when parameterized above the
Lov\'{a}sz-Plummer bound. The main result of this work is an
\(\OhStar(3^{\hat{k}})\) algorithm for \VCALP.

The presence of both \((-2LP(G))\) and \(MM(G)\)---in addition to
the ``solution size''---in our measure made it challenging to find
structures on which to branch; we had to be able to control each
of these values and their interplay in order to ensure a drop in
the measure at each branch. The main new idea which we employed
for overcoming this hurdle is the use of the Gallai-Edmonds
decomposition of graphs for finding structures on which to branch
profitably. The main technical effort in this work has been
expended in proving that our choice of vertices/edges from the
Gallai-Edmonds decomposition actually work in the way we
want. Note, however, that the branching rules themselves are very
simple; it is only the analysis which is involved.

The most immediate open problem is whether we can improve on the
base \(3\) of the \FPT running time.  Note that any such
improvement directly implies \FPT algorithms of the same running
time for \AGVC and \VCAL. Tempted by this implication, we have
tried to bring this number down but, so far, in vain. Another
question which suggests itself is: Is this the best lower bound
for vertex cover number above which \VC is \FPT? How far can we
push the lower bound before the problem becomes intractable?

%%% Local Variables: 
%%% mode: latex
%%% TeX-master: "main"
%%% End: 

%% file: appendix.tex
%%% Local Variables: 
%%% mode: latex
%%% TeX-master: "main"
%%% End: 